\documentclass[times, 10pt,twocolumn]{article} 
\pdfoutput=1
\usepackage{latex8}
\usepackage{times}
\usepackage[latin2]{inputenc}

\usepackage{amssymb}
\usepackage{amsmath}
\usepackage{comment}
\usepackage{graphicx}

\author{
  Eryk~Kopczyński\thanks{Supported by the Polish government grant~no.~N206~008~32/0810}
  \\
  Institute of Informatics, Warsaw University \\
  \texttt{erykk@mimuw.edu.pl}
}

\title{Complexity of Problems for Commutative Grammars}

\newcounter{mycount}[section]

\newtheorem{fact}[mycount]{Proposition}
\newtheorem{lemma}[mycount]{Lemma}
\newtheorem{definition}[mycount]{Definition}
\newtheorem{theorem}[mycount]{Theorem}
\newtheorem{corol}[mycount]{Corollary}

\newenvironment{proof}[1][]%
{\medskip{\bf Proof#1. }}%
{\par\medskip}

\def\squareforqed{\hbox{\rlap{$\sqcap$}$\sqcup$}}
\def\qed{\ifmmode\squareforqed\else{\unskip\nobreak\hfil
\penalty50\hskip1em\null\nobreak\hfil\squareforqed
\parfillskip=0pt\finalhyphendemerits=0\endgraf}\fi}

\def\ra{\rightarrow}

\def\supp{{\rm{supp}}}
\def\bbN{{\mathbb N}}
\def\bbR{{\mathbb R}}
\def\bbQ{{\mathbb Q}}
\def\bbQp{{\mathbb P}}
\def\bbQpo{{\mathbb P}}
\def\bbP{{\mathbb P}}
\def\bbZ{{\mathbb Z}}

\def\mult#1{{\bbN^{#1}}}

\def\source{{\rm source}}
\def\target{{\rm target}}
\def\angle{{\rm angle}}
\def\reg{{\rm reg}}
\def\Reg{{\rm Reg}}
\def\prod{{\rm out}}
\def\sgn{{\rm sgn}}
\def\int{{\rm int}}

\def\mv#1#2#3{#1\stackrel{#2}{\rightarrow}#3}

\def\alph{\Sigma}
\def\alphs{{|\Sigma|}}
\def\stats{{|S|}}
\def\Cyc{\mathcal{C}}
\def\OCyc{\mathcal{Y}}
\def\PiP{$\mathrm{\Pi_2^P}$}
\def\bfPiP{$\mathbf{\Pi_2^P}$}
\def\SigP{$\mathrm{\Sigma_2^P}$}

\def\regf{\mathcal{R}}

\def\osum#1#2{#2^{\oplus #1}}

\def\maxcycle{{C_1}}
\def\maxrunp{{C_2}}
\def\maxrun{{C_3}}
\def\cyclemul{{C_4}}
\def\macco{{C_5}}
\def\rgden{{C_6}}
\def\reggen{{C_7}}
\def\regorder{{C_{8}}}
\def\csep{{C_9}}
\def\mulgen{{C_{10}}}
\def\bigconst{{C_{11}}}
\def\regval{\maxrun}

\begin{document}
\maketitle
\pagestyle{plain}

\begin{abstract}
We consider Parikh images of languages accepted by non-deterministic finite
automata and context-free grammars; in other words, we treat the languages in
a commutative way --- we do not care about the order of letters in the
accepted word, but rather how many times each one of them appears. In most
cases we assume that the alphabet is of fixed size. We show tight complexity
bounds for problems like membership, equivalence, and disjointness.
In particular, we show polynomial algorithms for membership and disjointness
for Parikh images of non-deterministic finite automata over fixed alphabet,
and we show that equivalence is \PiP complete for context-free grammars over
fixed terminal alphabet.
\end{abstract}

\section{Introduction}

We consider languages accepted by regular and context-free grammars,
except that we treat the language in a commutative way ---
we do not care about the order of letters in
the accepted word, but rather how many times each one of them appears. 
In this
setting, usual problems, like membership and equivalence, have different
complexities than in the non-commutative case.

A well known classic result in this area is the result of Parikh \cite{parikh}
that, for a context-free grammar $G$ over alphabet $\alph$, the Parikh image
of $G$, i.e., the set $\prod(G) \subseteq \bbN^\alph$ of such multisets $M$
that, in some word $w \in L(G)$, each letter $x$ appears $M(x)$ times, is a
semilinear set. Some complexity results regarding semilinear sets and
commutative grammars have been obtained by
D. Huynh \cite{semipi,semipi2}, who has shown that equivalence is \PiP-hard
both for semilinear sets and commutative grammars (where \PiP\ is the dual of
the second level of the polynomial-time hierarchy, \cite{polyh}).

There are many practical uses of regular and context-free languages which
do not care about the order of the letters in the word. For example, 
when considering regular languages of trees, we might be not 
interested in the ordering of children of a given node. \cite{xml} and
\cite{xml2} consider XML schemas allowing marking some nodes as unordered.

Some research has also been done in the field of communication-free
Petri nets, or Basic Parallel Processes (BPP). A Petri net 
(\cite{petri1}, \cite{petri2}) is communication-free if each transition has
only one input. This restriction means that such a Petri net is essentially
equivalent to a commutative context-free grammar. \cite{hcyen} shows that
the reachability equivalence problem for BPP-nets can be solved in
$DTIME\left(2^{2^{ds^3}}\right)$. For general Petri nets, reachability
(membership in terms of grammars) is decidable
\cite{kosaraju}, although the known algorithms
require non-primitive recursive space; and reachability equivalence
is undecidable \cite{hack26}. Also, some harder types of equivalence problems
are undecidable for BPP nets \cite{huttel}. See \cite{esparza2} for a survey
of decidability results regarding Petri nets.

It turns out that, contrary to the non-commutative case, the size of alphabet
is very important. In the
non-commutative case, we can use strings $a$, $ab$, and $abb$
to encode a three letter alphabet $\{a,b,c\}$ using two letters. Trying to
do this in the commutative case fails, since two different words $ac$ and $bb$
are mapped to $aabb$ and $abab$, which are commutatively the same word.
There is no way to map a three letter alphabet to a two letter one
which does not collapse anything. Each new letter adds a new dimension to
the problem in the commutative case --- literally: 
commutative words (multisets) over an alphabet of size $d$ are better viewed
as points in a $d$-dimensional space, rather than strings.

Contrary to most previous papers on commutative grammars, in most cases we
assume that our (terminal) alphabet is of fixed size. As far as we know, there
have been no successful previous attempts in this direction (except for the
much simpler case $d=1$ \cite{hyunh1}). Our methods enable us to obtain tight
complexity bounds for most of the basic problems (like membership, inclusion,
equivalence, universality, disjointness) for both regular and context-free
commutative grammars, over an alphabet of fixed size.
In some cases, we provide algorithms for the special case $d=2$, as they are
much simpler than the general ones.

In Theorem \ref{algoregmem}, we show a polynomial
algorithm deciding membership for regular languages, i.e., whether a
multiset (given as a binary representation) is in the Parikh image of
a regular language (represented by the non-deterministic finite automaton
accepting it).

In Theorem \ref{grammarnormalform}, we improve upon Parikh's
result quoted above in two ways, assuming that the alphabet is of size 2. First,
$\prod(G)$ is produced as a union of linear sets with only two periods
(whose magnitude is single exponential in size of $G$);
second, these linear sets can be grouped in a
polynomial number of bundles such that each bundle shares the pairs of
periods used (we call such a bundle an $A,B$-frame). This leads to
a \PiP algorithm for solving inclusion (equivalence) for context-free
languages. Unfortunately, such simple presentation is impossible for
alphabets of size greater than 2; we provide a counterexample where
$d=3$, and a much more complicated reasoning which solves the general case
(still resulting in a \PiP algorithm).

The following table summarizes our results.
Alphabet size F means that alphabet is of fixed size, and
U means unfixed size. We consider the basic problems: membership, 
universality, inclusion, and disjointness; note that solving inclusion is
equivalent to solving equivalence --- simple reductions exist in
both ways. We use c as an abbreviation for complete. Our main results ---
our most important algorithms ---  are marked with bold (polynomial algorithms
for checking membership and disjointness for regular grammars over alphabets
of fixed size, \PiP-completeness of the inclusion (equivalence) problems
for context-free grammars over alphabets of fixed size). Problems which have
been shown to be hard are marked with stars (NP-completeness of membership
checking for regular grammars over alphabets of unfixed size and context-free
grammars over alphabets of size 1, coNP-completeness of universality checking
for regular grammars over alphabets of size 1, \PiP-completeness of inclusion
(equivalence) checking for context-free grammars); the proofs are simple and
have been included for sake of completeness.

\begin{center}
\begin{tabular}{|c|cccc|}
 \hline
\multicolumn{5}{|c|}{regular languages} \\
 \hline
alphabet size & 1 & 2 & F & U                \\
 \hline
membership    & P & P & {\bf P} & NPc*       \\
universality  & coNPc* & coNPc & coNPc & ?   \\
inclusion     & coNPc & coNPc & coNPc & ?    \\
disjointness  & P  & P  & {\bf P}  & coNPc   \\
\hline
\multicolumn{5}{|c|}{context-free languages} \\
 \hline                                              
alphabet size & 1 & 2 & F & U \\                     
 \hline                                              
membership    & NPc* & NPc & NPc & NPc \\
universality  & \PiP & \PiP & \PiP & ? \\
inclusion     & \PiP c & \bfPiP {\bf c} & \bfPiP {\bf c} & ? \\
disjointness  & coNPc & coNPc & coNPc & ? \\
\hline
\end{tabular}
\end{center}

\section{Overview}

In this section, we present our techniques and results in an informal way.
The formal version can be found in the following sections.

Our main observation is that we can treat our runs
(or derivations in CF grammars) completely commutatively: we just count
how many times each transition (rule) has been used. In both cases, the
validity of such a ,,commutative run'' can be checked by checking two very
simple conditions: Euler condition (each state is entered as many times
as it is used) and connectedness (there are no unconnected loops) ---
Theorem \ref{algoder}. From this,
we immediately get that checking membership of a given multiset in a Parikh image
of a context-free language is in NP.

The second observation is that we can decompose a run into smaller parts,
ultimately obtaining its \emph{skeleton}, to which we add some (simple) \emph{cycles}.
Since the skeleton uses all states that the original run used (which we call
its \emph{support}), we can add these cycles in arbitrary numbers to our
skeleton, always getting valid runs. Moreover, in case of finite automata
(regular grammars), both the skeleton and the cycles are bounded polynomially
(in the size of the automaton, $|G|$ --- Lemma \ref{fatedge}).

Now, linear algebraic considerations come into play. Whenever we have 
$d$ linearly independent vectors $v_1, \ldots, v_d$ with integer coordinates in a $d$-dimensional
space, for each other vector $v$, there is a constant $C$ such that $Cv$ can
be written as a linear combination of $v_1, \ldots, v_d$ with integer
coefficients (Lemma \ref{detlemma}). This $C$ is bounded polynomially by
coordinates of $v_1, \ldots, v_d$ ($d$ appears in the exponent). In our
case, our vectors will be the Parikh images of our cycles, with $d$ letters
in our alphabet.

Thus, whenever we have a non-negative integer combination of more than $d$
cycles, where the multiplicities of cycles are big enough, we can reconstruct
our Parikh image using different multiplicities of these cycles, and do such
,,shifting'' until
the multiplicities of some cycles drop (Lemma \ref{wrilemma} and Theorem
\ref{theoreg}). Thus, there are at most $d$ cycles which we are using in
large quantities. From this, we get an algorithm for membership
(Theorem \ref{algoregmem}): we can guess the small run and the $d$ cycles
(making sure that the run crosses these cycles), and then just check whether
we obtain our result by adding the cycles in non-negative integer amounts to
our run --- which boils down to solving a system of equations. This algorithm
is polynomial, since everything it uses is bounded polynomially.

The situation is the simplest in the case of $d=2$, where instead of guessing
the two cycles, we can always use the two extreme ones $v_a, v_b$ --- i.e.,
the ones which proportionally contain the greatest quantities of the two
letters $a$ and $b$ in our alphabet. Each other cycle can be written as
a non-negative combination of these two. Still, we have to take the extreme
cycles which cross our small run --- which means that we have to guess the
two states where they cross, and then take extreme cycles crossing these
states. For unfixed $d$, or for context-free grammars,
the problem is NP-complete.

Now, what about context-free grammars? Generally, we can use the same
techniques, but now, skeletons and cycles are bounded exponentially.
In case of $d = 2$, we get Theorem \ref{grammarnormalform}: a Parikh
image of $G$ is a union of a polynomial number of \emph{$A,B$-frames};
where an $A,B$-frame (Definition \ref{defab}) is a set of vectors defined
by some $W$ (a subset of $\bbN^d$ bounded by $A$) and two vectors $v_a, v_b$
(bounded by $B$),
consisting of all vectors of form $w+n_a v_a + n_b v_b$ (where $w\in W$).
The number of $A,B$-frames is polynomial, because our two vectors will
always correspond to extreme cycles from some two states. $A$ and $B$ are
exponential.

The following
picture shows geometrically what an $A,B$-frame is: a set $W$ sitting
inside of a box of size $A$ (drawn as the letter $W$) is copied by shifting
it in two directions. The vectors by
which we are shifting are bounded by $B$.

\begin{center}
\includegraphics{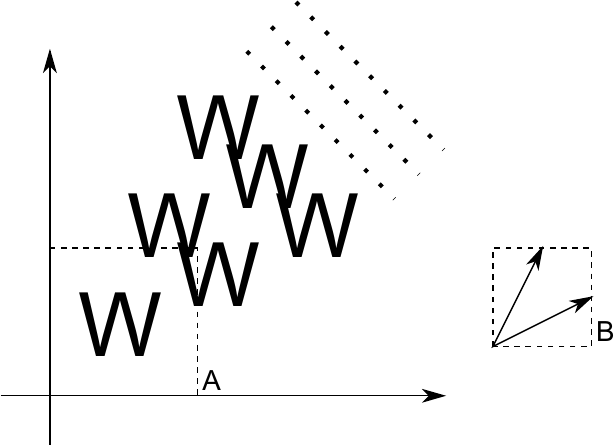}
\end{center}

It turns out that two such unions of $A,B$-frames are equal iff
they are equal in the exponentially bounded region close to 0
(Lemma \ref{anglelemma} and Lemma \ref{formcollapse}).
Together with the fact that membership checking is in NP, we get a \PiP
algorithm for checking inclusion (equivalence) of Parikh images of context-free
grammars (for $d=2$).

For $d>2$, it may be impossible to get a polynomial number of $A,B$-frames
(a nice counterexample in Section \ref{secmore}), which means that Lemma
\ref{formcollapse} fails (the region would be bounded double exponentially).
However, we can circumvent this by splitting $\bbN^d$ into \emph{regions} ---
when restricted to a single region, the number of $A,B$-frames will be
polynomial, thus allowing us to use Lemma \ref{formcollapse} successfully in 
each region separately, again getting a \PiP\ algorithm for deciding equivalence.



\section{Geometry of Multisets}

For a set $X$, the elements of $\mult X$ are interpreted as multisets of
elements of $X$. A set is interpreted as a multiset, thus for $Y \subseteq X$,
$Y(x) = 1$ for $x \in Y$ and $Y(x) = 0$ otherwise. For $x \in X$ and $v \in
\mult X$, $x \in v$ denotes $v(x) > 0$. We will sometimes write
$x$ instead of $\{x\}$ (a multiset containing only a single occurence of
$x \in X$). $\mult X$ is also treated as a subset of $\bbZ^X$, $\bbQ^X$,
and $\bbR^X$.
For $v \in \bbQ^X$, $|v| = \sum_{x\in X} |v(x)|$, $||v|| = 
\max_{x\in X} |v(x)|$. We say $u \geq v$ iff $u_x \geq v_x$ for each $x$.

By $F^X_X$ we denote the set of matrices with coefficients in $F$ and
dimensions indexed with elements of $X$: for a matrix $M \in F^X_X$, $M^i$,
the $i$-th column of the matrix, is a vector in $F^X$. For $M \in \bbQ^X_X$
and $v \in \bbQ^X$, $Mv$ is a vector given by $(Mv)_j = \sum_i M^i_j v_i$,
and $||M|| = \max_{x\in X} ||M^x||$.

We use the notation $[0..K]$ for the set of integers from 0 to $K$,
$\bbQp$ for the set of non-negative rationals (we don't use the more
standard notation of $\bbQ^+$ to avoid double upper indexing, as in
$(\bbQ^+)^X$), and $[0;K]$ for the set
of rationals from 0 to $K$. Thus, for example, $[0..K]^X_X$ denotes matrices
with coefficients in $\bbN$ bounded by $K$.

We can add or multiply sets of scalars, vectors, or matrices, in the usual
way. For example, $U+V = \{u+v: u \in U, v \in V\}$, and $M \bbN^X$ for
$M \in \bbZ^X_X$ is the set of vectors which can be obtained as a linear
combination of columns of $M$ with coefficients from $\bbN$.

\begin{lemma}\label{detlemma}
Let $M$ be a non-degenerate matrix in $\bbZ^\alph_\alph$, and let $v \in \bbZ^\alph$.
Then $(\det M) v \in M \bbZ^\alph$.
\end{lemma}

\begin{proof}
For $v_1, v_2 \in \bbZ^\alph$, we say that $v_1 \equiv v_2$ iff
$v_1 - v_2 \in M \bbZ^\alph$.
The quotient group $\bbZ^\alph /_\equiv$ has $\det M$ elements (intuitively,
for $|\alph| = 2$, the number of elements is equal to the area of the
parallelogram given by columns of $M$; this intuition also works in other
dimensions). Thus, $(\det M) v \equiv 0$.
\qed\end{proof}

\begin{lemma}\label{wrilemma}
Let $V \subseteq [0..K]^\alph$ be a linearly dependent set of vectors.
Then for some $\alpha \in \bbZ^V$ we have $\sum_{v \in V} \alpha_v v = 0$,
where $|\alpha| = O(K^\alphs \alphs!)$, and $\alpha_v > 0$ for some $v$.
\end{lemma}

\begin{proof}
Without loss of generality we can assume that $V$ is a minimal
linearly dependent set.
Thus, we get $\sum_{v \in V} \beta_v v = 0$ for
some rational coefficients $\beta \in \bbQ^V$.
Let $u$ be such that $|\beta_u| \geq |\beta_v|$ for each $v$. Let
$M \in \bbZ^\alph_\alph$ be
a non-degenerate matrix whose $|V|-1$ columns are $V-\{u\}$
(we obtain a non-degenerate matrix since $V$ was a minimal linearly
dependent set; if $|V| < \alphs+1$, we fill up the remaining columns with
independent unit vectors).
From Lemma \ref{detlemma} we get that $(\det M) u = M w$ for some $w \in
\bbZ^\alph$. Let $\alpha_u = -\det M$, $\alpha_v = w_i$ where $v = M^i$,
and $\alpha_v = 0$ for remaining vectors. 
We have that $\sum \alpha_v v = 0$. Moreover, we have that 
for some $q \in \bbQ$ we have $\beta_v = q \alpha_v$ for each $v$
(for a minimal linearly dependent set, $(\beta_v)$ is unique up to a
constant); thus,
$|\alpha_v| < |\alpha_u| = \det M = O(K^\alphs \alphs!)$ for each $i$.
\qed\end{proof}

\begin{definition}\label{defab}
An {\bf $A,B$-frame} is a set of form $W + M \bbN^\alph$, where 
$W_i \subseteq [0..A]^\alph$, and $M_i \in [0..B]^\alph_\alph$. 
\end{definition}

\begin{lemma}\label{anglelemma}
Let $\alph = \{a, b\}$. Let $A$, $B$ and $C$ be positive integers.

For $w \in \bbQpo^\alph$ and $M$ in $[0..B]^\alph_\alph$,
let \[\angle(w,M) = w + M \bbQpo^\alph.\]

For $v \in \bbQpo^\alph$, the {\bf region} of $v$ is defined as
\[\reg(v) = \left\{u \in \bbQpo^\alph:
\begin{array}{l}
\forall w \in [0;A]^\alph \ \forall M \in [0..B]^\alph_\alph  \\
\ v \in \angle(w,M) \iff \\ u \in \angle(w,M)
\end{array}
\right\} \]

For each $v \in \bbN^\alph$.
there exists a vector $v' \in \bbN^\alph$ such that 
$||v'|| = O(AB^2+BC)$,
$v' \in \reg(v)$
and $v-v' \in C \bbZ^\alph$.
\end{lemma}

\begin{proof}
The following picture shows this lemma graphically for $B = 3$.
$A$ is the size of the black square in the bottom left corner. 
Lines shown on the picture are boundaries between angles; 
in each bundle, 6 lines are shown, but
it should be understood that there is actually a semi-line starting from
each rational point in the black square.

\begin{center}\includegraphics{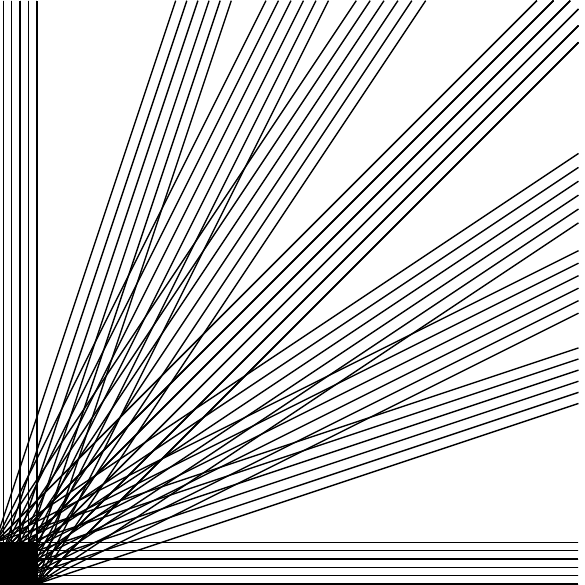}\end{center}

Each angle is the set of points between two semi-lines on the picture which
cross somewhere in the black square. There are three type of regions:
ones containing only one vector (each bounded region is actually a singleton),
8 unbounded regions in angular shapes (between two consective bundles of lines
--- there are 9 bundles of lines because vectors in $[0..3]^\alph$ go in
9 directions), and regions in shape of semi-lines.

It can be easily calculated that each point where bundles of lines going in
different directions cross has its coordinates bounded polynomially
(by $AB^2$). 

Let $v$ be a point. If $\reg(v)$ is a singleton, then we are done (because
$v$ is already bounded polynomially). Otherwise, $\reg(v)$ is the inside of
$\angle(w,M)$, where $w$ is bounded polynomially and $M$ is given by vectors
in two directions $M^1$ and $M^2$ (consecutive or equal).
If $v$ is inside the
paralellogram whose vertices are $w$, $w+CM^1$, $w+CM^2$, $w+CM^1+CM^2$,
then we are done (all those vertices are bounded polynomially). Otherwise
we subtract multiplicities of $CM^1$ and $CM^2$ until we get a point
$w'$ in this paralellogram.
\qed\end{proof}

\begin{lemma}\label{formcollapse}
Let $\alph = \{a, b\}$.
For $i \in I$, let $Z_i = W_i + M_i \bbN^\alph$ be an $A,B$-frame.
Let $v \in \bbN^\alph$. Then there exists a $v' \in \bbN^\alph$
such that $||v'|| = O((A+B)^{O(|I|)})$, and, for each $i$,
$v \in Z_i$ iff $v' \in Z_i$.
\end{lemma}

\begin{proof}
Assume that the matrices $M_i$ are non-degenerate (the case of
degenerate matrices can be solved easily by changing the matrices).

Let $C$ be the least common multiple of determinants of matrices
$M_i$, $C = O(B^{O(|I|)})$.

Let $v \in \bbN^\alph$. Let $v'$
be the vector $v'$ from Lemma \ref{anglelemma} for our $v$ and $C$;
we get that $||v'|| = O((A+B)^{O(|I|)})$.
We will show that it satisfies our conditions.

It is enough to check whether 
$v \in Z'$ iff $v' \in Z'$ for each $Z'$ of form $w_0 + M_i \bbN^\alph$,
where $w_0 \in W_i$.

Since $M_i$ is non-degenerate, for some $\alpha, \alpha' \in \bbQ^\alph$
we have $v = w_0 + M_i \alpha$ and $v' = w_0 + M_i \alpha'$. Since $v'$ is in the
same region as $v$, $\alpha \geq 0$ iff $\alpha' \geq 0$. On the other hand,
$v-v' \in C \bbZ^\alph \subseteq (\det M_i) \bbZ^\alph \subseteq
M_i \bbZ^\alph$ from Lemma \ref{detlemma}. Thus, $v \in \bbN^\alph$ iff
$v' \in \bbN^\alph$.
\qed\end{proof}

\section{Commutative Grammars}

Since in this paper we don't care about the order of symbols in strings
generated by our grammars, we define our grammars commutatively: a state
(nonterminal) produces a multiset of letters and states, not a string.

Derivation trees are defined for commutative grammars similarly as for
the usual ones; we omit this definition. However, we usually also abstract
from derivation trees, by considering our runs as multisets rather than trees:
we don't care where in the tree each transition (production) has been used,
we just count the total number of occurences. We show that there is a simple
condition which checks whether our multiset corresponds to some full derivation,
or a ,,cyclic'' derivation. (A similar algebraic definition of cycles is used
by the algebraic topologists.)

A {\bf commutative grammar} is a tuple $G = (\alph, S, s_0, \delta)$, where
$\alph$ is a finite {\bf alphabet}, $S$ is a finite set of {\bf states}, 
$s_0 \in S$ is an initial state, and $\delta \subseteq S \times \mult A \times
\mult S$ is a set of {\bf transitions}. We will write transitions $(s,a,t)$
as $\mv{s}{a}{t}$; in terms of derivations, each transition consumes the state $s$ and
produces each letter from $a$ and each state from $t$. For a transition
$\tau = \mv{s}{a}{t}$, $\source({\tau}) = s$, $\target({\tau}) = t$, and
$\prod({\tau}) = a$. 

We will assume that each state is a source of some transition.
We will also assume that for each $\mv{s}{a}{t} \in \delta$, $|a| \leq 1$ and
$|t| \leq 2$. (We do this because we want to limit things produced by the
grammar in terms of $\stats$. Grammars not satisfying these conditions can be
easily transformed by adding additional states.)
A commutative grammar satisfying
$|t| \leq 1$ is called a {\bf regular commutative grammar}
(regular grammars are equivalent to non-deterministic finite automata,
with initial state $s_0$ and transitions with $\target(\tau) = 0$ as
transitions to the final state; we prefer to speak about regular grammars 
rather than NFAs for the sake of uniformness).

For a $D \in \mult \delta$, $\source(D) \in \mult S$ counts how many
each state appears as source of a transition: $\source(D) (s) = 
\sum_{\tau: \source(\tau) = s} D(\tau)$, and $\prod(D)$ and $\target(D)$
counts how many each letter and each state, respectively, is produced:
$\prod(D) = \sum_\tau D(\tau) \prod(\tau)$, $\target(D) = \sum_\tau D(\tau)
\target(\tau)$.  The {\bf support} of $D$,
$\supp(D) = \{s \in S: s \in \source(D)\}$. We say that $D$ is {\bf connected
from $s \in S$} if for each $t \in \supp(D)$ there is a path from $s$ to $t$
in $D$, i.e., a sequence $\tau_1, \ldots, \tau_m$ such that $\tau_i \in D$,
$s  = \source(\tau_1)$, $\source(\tau_{i+1}) \in \target(\tau_i)$,
$t \in \target(\tau_m)$. We say that $D$ is a {\bf cycle from $s \in S$} iff
it is connected from $s$ and it satisfies the {\bf Euler condition}:
$\source(D) = \target(D)$ (in terms of derivations, each state is consumed as
many times as it is produced).  We say that $D$ is {\bf run} iff it is connected from
$s_0$ and $\source(D) = \target(D) + \{s_0\}$ (each state is consumed as
many times as it is produced, except $s_0$ which is consumed one time
more).

For a commutative grammar $G$, $\prod(G) = \{\prod(D): D\mbox{ is a run in }G\}$.

The relation between algebraic runs and cycles and derivation trees is as follows:

\begin{fact}\label{algoder}
Let $G$ be a commutative grammar. Then:
\begin{itemize}
\item $D$ is a run iff there is a derivation tree from $s_0$ where each
transition $\tau$ appears $D(\tau)$ times, and all the branches are closed,

\item $D$ is a cycle from $s$ iff there is a derivation tree from $s$ where
each transition $\tau$ appears $D(\tau)$ times, and all the branches are closed
except one with state $s$ at its end (we call such derivation tree {\bf cyclic}).
\end{itemize}
\end{fact}

\begin{proof}[ of Proposition \ref{algoder}]
We show the proof for runs (for
cycles the proof is similar).

Start with $s_0$ and try applying transitions from $D$ 
(obviously, using each transition as many times as it appears in $D$)
as long as we have some open branches. If we have
used all the elements of $D$ in the process, we are done. Otherwise, 
since the run $D$ is connected, there
must be some state $s$ such that $D$ contains some transition from $s$ which
is still not used, and $s$ already appears in our derivation constructed
so far. Since each derivation
uses each state as many times as it was produced, and so does $D$, there
also must be a yet unused transition $\tau_1$ in $D$ which produces $s$,
from, say, $s_1$. If $s \neq s_1$, for the same reason there must be a yet
unused transition $\tau_2$ in $D$ which produces $s_1$ from some $s_2$.
Finally, we produce some $s_k$ from $s$. We create a cyclic 
derivation tree with transitions $\tau_k, \ldots, \tau_1$ on its
main branch, closing all the side branches with remaining unused transitions
from $D$. We insert this cycle into our tree
(we have produced $s$ in some place; we cut off the part of tree
from this $s$, insert our cycle here, and we attach the part of tree we
cut off to the open branch of our cycle).
Repeat until all elements of $D$ have been used.
\qed\end{proof}

Thus, if $G$ is a commutative version of some context-free grammar $H$,
then $\prod(G)$ equals the Parikh image of $L(H)$, i.e., $v \in \prod(G)$ iff
there exists a $w \in L(H)$ such that each letter $a$ appears in $w$ $v(a)$
times.

One inclusion is obvious. In the case of regular grammars,
the cycle is just what is expected (a cycle in the transition graph),
and the other inclusion is equivalent to the classic theorem of Euler
(characterization of graphs with Eulerian paths and cycles);
in general, it is a simple generalization. 

A cycle is called a {\bf simple cycle} iff it cannot be decomposed 
as a sum of smaller non-zero cycles, and a run $D$ is called a {\bf skeleton run} 
if it cannot be decomposed as a sum of a run $D_1$ and a non-zero cycle $C$,
where $\supp(D_1) = \supp(D)$.
For each state $s \in S$, let $\Cyc_s$ be the set of simple cycles
from $s$, and $\Cyc_T = \bigcup_{s\in T} \Cyc_s$, for $T \subseteq S$.
Also, let $\OCyc_s = \prod(\Cyc_s)$, and 
$\OCyc_T = \bigcup_{s\in T} \OCyc_s$ (cycle outputs).

\section{Membership checking}

\begin{lemma}\label{fatedge}
Let $G$ be a regular commutative grammar, and $D$ be a run such that
$|D| > 1+n|\delta|$. Then $D = D_1 + n C$, where
$D_1$ is a run, $C$ is a simple cycle, and $\supp(D_1) = \supp(D)$.
\end{lemma}

In case of $n=1$, we get a limit on the size of a skeleton run.

\begin{proof}
Let $\tau$ be a transition such that $D(\tau) \geq 1 + n |\delta|$.
We have $|\target({\tau})| = 1$ (it cannot be greater because
$G$ is regular, and cannot be 0 becuase each run in a regular grammar has
exactly one transition with $|\target({\tau})| = 0$).
Let $t \in \target({\tau})$. Let $T$ be the set of
states which can be reached from $t$ via a path using only transitions
$\tau_i$ such that $D(\tau_i) > n$. If $\source(\tau) \in T$, this finishes
the proof (we have found a cycle in the graph of transitions, which can be
easily translated to an algebraic cycle). Otherwise, let $u = \sum_{\tau \in D:
\source(\tau) \notin T, \target(\tau) \in T} D(\tau)$, and
$v = \sum_{\tau \in D: \source(\tau) \notin T, \target(\tau) \in T} D(\tau)$.
From the Euler condition, we get that $u = v$. Since $\tau$ is counted in 
$u$ $1+ n |\delta|$ times, and there are $|\delta|$ transitions, there
must exist a transition $\tau'$ which is counted $n+1$ times in $v$. This
is a contradiction (we have found a path from $t$ to $\target(\tau') \notin T$).
\qed\end{proof}

\begin{theorem}\label{theoreg}
Let $G$ be a regular commutative grammar, and $K \in \mult \alph$. Then
$K \in \prod(G)$ iff there exists a run $D$ in $G$, $||D|| = 
O(\stats^{2\alphs} \alphs!)$, and
simple cycles $C_1, \ldots, C_m$, $C_i \in \Cyc_{\supp(D)}$, such that
$C_i$ are linearly independent and $K = \prod(D) + \sum_i \alpha^i \prod(C_i)$
for some $\alpha^1, \ldots, \alpha^m$.
\end{theorem}

\begin{proof}

\def\ocdo{\OCyc_{\supp(D_0)}}

Let $D_0$ be a run in $G$ such that $K = \prod(D_0)$.
We decompose the run $D_0$ into a sum of simpler runs (on the same
support) and simple cycles,
until we get $D_0 = D_2 + \sum_{C \in \Cyc_{\supp(D_0)}} \gamma_C C$, where
$D_2$ is a skeleton. From Lemma \ref{fatedge} we get that
$||D_2|| \leq |\delta|$. By taking $\prod$'s, we get
$K = \prod(D_2) + \sum_{Y \in \ocdo} \gamma_Y Y$, where $\gamma_Y \in \bbN$
for each $Y$.

Let $P \subseteq \ocdo$ be the set of such cycle outputs $Y$ that
$\gamma_Y \geq L$ for $L = O(\stats^\alphs \alphs!)$.
We can decompose $K$ so that $P$ is linearly independent.
Otherwise, by Lemma \ref{wrilemma}, 
$\sum_{Y\in P} \alpha_Y Y = 0$ for
some $\alpha_Y$, $|\alpha_Y| \leq L$, and $\alpha_Y > 0$ for some $Y$.
This allows us to transfer multiplicites between different cycles: 
if we take $\gamma'_Y = \gamma_Y - \alpha_Y$ for $Y \in P$, and
$\gamma'_Y = \gamma_Y$ for $Y \notin P$, we have
$\sum_Y \gamma_Y Y = \sum_C \gamma'_C Y$. We transfer
multiplicites (i.e., replace $\gamma$ with $\gamma'$) until one of our cycles
is no longer in $P$.

Let $D_1 = D_2 + \sum_{C \in \Cyc_{\supp(D_1)} - P} \gamma_C C$. Since
$||D_2|| \leq |\delta|$, $\gamma_C < L$, and there are at most
$O(\stats^\alphs)$ distinct simple cycles in $\ocdo$ up to equivalence of
$\prod$'s, we get that $||D_1|| \leq |\delta| + L O(\stats^\alphs)$. 
Now, $K = D_1 + \sum_{Y \in P} \gamma_Y Y$.
\qed\end{proof}

\begin{theorem}\label{algoregmem}
For an alphabet $\alph$ of fixed size, and a commutative regular grammar 
$G$ over $\alph$, and $K \in \mult \alph$,
the problem of deciding whether $K \in \prod(G)$ is in P.
\end{theorem}

\begin{proof}
The theorem \ref{theoreg} remains true if we define $\Cyc_s$ and 
$\OCyc_s$ using 
\emph{short cycles} instead of simple cycles --- a cycle $C$ is short
iff $|C| \leq |S|$. 
This allows us to calculate sets $\OCyc_s$ for each state $S$
using simple dynamic programming.

For each $T \subseteq S$ of size at most
$\alphs$, we calculate the set of possible $\prod(D)$ with $D$
satisfying the limit from Theorem \ref{theoreg} and $T \subseteq \supp(D)$.
For each element
of $\prod(D)$ and each sequence of linearly independent elements of $\OCyc_T$,
$Y_1\ldots Y_m$, we check whether $K = \prod(D) + \sum_i \alpha^i Y_i$ for
some $\alpha^i$, which can be done by solving a system of equations.
\qed\end{proof}

\begin{theorem}\label{algoregdis}
For an alphabet $\alph$ of fixed size, and two commutative regular grammars 
$G$ and $H$ over $\alph$, the problem of deciding whether 
$\prod(G) \cap \prod(H) = \emptyset$ is in P.
\end{theorem}

\begin{proof}
Note that in the proof of Theorem \ref{algoregmem} we have actually never
used our assumption that outputs of our transitions are non-negative, e.g.,
Lemma \ref{wrilemma} works as well for $V \subseteq [-K..K]^\Sigma$. Thus, we
can check whether $\prod(G)$ and $\prod(H)$ are disjoint by checking whether
$0 \notin \prod(GH^{-1})$, where $H^{-1}$ is obtained from $H$ by negating
outputs of all transitions, and $GH^{-1}$ is a regular grammar obtained via
the usual method of concatenating languages given by regular grammars $G$
and $H^{-1}$.
\qed\end{proof}

\section{Inclusion checking}

\def\expon{O\left(2^{|S|^{O(1)}}\right)}

\begin{lemma}\label{exponlem}
Let $G$ be a commutative grammar over $\alph$.
If $D$ is a simple cycle or a skeleton run, then $\prod(D) = \expon$.
\end{lemma}

\begin{proof}
We start with the cycle case. We consider its cyclic derivation tree
from Proposition \ref{algoder}.

If somewhere on the branch leading to $s$ (the main branch) we had
another $s$, we can easily split our cycle into two cycles (by splitting
the derivation tree). A similar thing can be done if we had some state
$t$ in two places on the main branch.

A similar operation can be done when we find the same state twice
on the side part of a branch (i.e. the part disjoint with the main branch).

Since we can use each state at most twice on each branch (once on the
main part and once on the side part), this limits the
size of a simple cycle to exponential in size of $G$.

The construction for skeletons is similar. Indeed, 
consider a skeleton run. If a state $s$ appears $|S|+1$ times on
a branch of a production tree, it means that there exist two 
consecutive appearances
of $s$ such that the part of tree between them can be cut off without
removing any state from the support of this skeleton (otherwise
each such state would have to be different and we would have $|S|+1$ states
in total).
\qed\end{proof}

\begin{theorem}[``normal form'']\label{grammarnormalform}
Let $\alph = \{a, b\}$, and $G$ be a commutative grammar over $\alph$.
Then $\prod(G) = \sum_{i\in I} Z_i$, where 
$|I| = |S|^{O(1)}$, and $Z_i$ are $A,B$-frames,
where $A , B = \expon$.
\end{theorem}

\begin{proof}

Let $D$ be a run of $G$, and $T = \supp(D)$. For $Y \in \OCyc_T$,
let $b(Y) = Y(b) / |Y|$; let $Y^a$ and $Y^b$ be the elements of
$\OCyc_T$ with the smallest and largest $b(Y)$, 
respectively. We have
$|Y^a|, |Y^b| = \expon$ from Lemma \ref{exponlem}.
There are at most $|S|^2$ possible pairs $(Y^a, Y^b)$. Let
$R(Y^a, Y^b)$ be the set of runs having particular $Y^a$ and $Y^b$.
We will show that $\prod(R(Y^a, Y^b))$ is of form 
$W + M \bbN^\alph$, where the columns of $M$ are $Y^a$ and $Y^b$.

We decompose $D$ as $D_2 + \sum_{C \in \Cyc_T} \alpha_C C$,
where $D_2$ is a skeleton. Thus, $\prod(D)$ is decomposed as
$\prod(D_2) + \sum_{Y \in \OCyc_T} \alpha_Y Y + \beta_a Y^a +
\beta_b Y^b$, where $\beta_a, \beta_b, \alpha_Y \in \bbN$.
We can assume that each $\alpha_Y < \det M$, because otherwise we
can replace $(\det M) Y$ by $q_a Y^a + q_b Y^b$, where $q_a, q_b \in \bbN$
(the coefficients are integers from Lemma \ref{detlemma} and non-negative
since $Y^a$ and $Y^b$ are extreme cycles). Each $Y$ and $\prod(D_2)$ is
$\expon$ from Lemma \ref{exponlem}, and there
are $\expon$ possible $Y$'s, thus 
$D_3 = \prod(D_2) + \sum_{c \in \prod(\Cyc_T)} \alpha^c c$ satisfies
$||D_3|| = \expon$.

By taking for $W$ the sets of possible $D_3$ for all runs from
$R(Y^a, Y^b)$, we get the required conclusion.
\qed\end{proof}

\begin{theorem}\label{maintheorem}
Let $G_1$ and $G_2$ be two commutative grammars over $\alph = \{a, b\}$.
Then the problem of deciding $\prod(G_1) \subseteq \prod(G_2)$ is
\PiP-complete.
\end{theorem}

\begin{proof}
The problem is \PiP-hard because we can reduce the problem of
semilinear set inclusion \cite{semipi} to it.

Using Theorem \ref{grammarnormalform}, we can write each
$\prod(G_k)$ as $\bigcup_{i \in I_k} Z^k_i$, where $I_k$ is a polynomial
set of indices and $Z^k_i$ is a $A,B$-frame, where $A$ and $B$ are
$\expon$.

From Lemma \ref{formcollapse} we get that it is enough to check inclusion
on vectors $v \in \bbN^\alph$ of size $|v| < P = O((A+B)^{O(|I_1|+|I_2|)})$.
We call such vectors small vectors.

A witness for membership of $v$ in a grammar $G$ is a run
$D$ such that $\prod(D) = v$, and $|v| < P$. If $v$ is small, 
and $D$ does not contain non-productive cycles (i.e., $C$ such that
$\prod(C) = 0$; such cycles can be eliminated),
then it can be described as a string of length polynomial in size
of $G$. We call such witness a small witness.

For each $v$, and each small witness of membership of $v$ in $G_1$, we have
to find a small witness of membership of $v$ in $G_2$. This can be done
in \PiP.
\qed\end{proof}

\section{Normal form over larger alphabets?}\label{secmore}

In Theorem \ref{maintheorem} we assumed that we are working with an
alphabet of two letters. Does a similar statement hold for alphabets of
size 3, 4, $\ldots$? What about alphabets of unfixed size?

For 2 letters, we have generated all multisets generated by our grammar
from runs $D$ having specific $\supp(D)$ using two extreme cycles, which
led to generating $\prod(D)$ using $N^2$ pairs of extreme cycles in total
 --- Theorem
\ref{grammarnormalform}. A natural
conjecture is that a similar normal form exists for greater alphabets,
except that there would be a polynomial ($N^d$) number of extreme cycles now
--- this would give us a straightforward generalization of Theorem
\ref{grammarnormalform}, and thus also of Theorem \ref{maintheorem},
by combining with a generalization of Lemma \ref{formcollapse}. However,
this is not true; in fact, Theorem \ref{grammarnormalform} already fails
for a three letter
alphabet. We present this counterexample, because we think it is interesting.

\begin{theorem}\label{counterthree}
There exists a context-free grammar $G$ over $\{x,y,z\}$ such that
$\prod(G)$ is not a union of a polynomial number of $A,B$-frames.
\end{theorem}

\begin{proof}
Consider the following grammar $G$
(in the standard commutative grammar notation, with exponential restrictions on
the size of productions):

\[
\begin{array}{cclcl}
S & \ra & 0 & | & SABCDEz \\
A & \ra & B^2 & | & C^2 D^2 E^2 x y \\
B & \ra & C^2 & | & D^2 E^2 x^2 y \\
C & \ra & D^2 & | & E^2 x^4 y \\
D & \ra & E^2 & | & x^8 y \\
E & \ra & 0   & | & x^{16} y
\end{array}
\]

The state $S$ generates any number of $z$'s together with the same number
of $ABCDE$'s. $ABCDE$ generates a convex 32-gon on the surface $\bbN^{\{x,y\}}$
(we get 32 corners by deciding which transition always
to use for each of five states $A$, $B$, $C$, $D$, $E$; they are points
with coordinates $({\frac{y(y+1)}{2}},y)$ for $y \in [0..31]$).
Since we generate $z^n$ together with $(ABCDE)^n$,
$\prod(G)$ is a cone (i.e., a unbounded pyramid) with 32 edges (each edge is the line
$\{(\frac{zy(y+1)}{2}, zy, z\}): z \in \bbR\}$ for some $y$),
and hence we need more than 16 three-dimensional $A,B$-frames to cover
$\prod(G)$.
This example generalizes to any number of states (bigger examples are
constructed using the same simple rule as the example above) --- we need more
than $2^{n-2}$ $A,B$-frames for a grammar with $n$ states and two transitions
for each state. Note that the $n$-state verison of the grammar above can
be written in the limited form (i.e., for each derivation $\mv{s}{a}{t}$,
$|a|\leq 1$, $|t|\leq 2$) using $O(n)$ states.

Instead of proving that this construction gives a good counterexample for
each $n$ (i.e., it indeed generates a $2^n$-gon), 
we present another construction, based on the same idea (although we
don't get as beatiful grammar as above, the proof is simpler).

For each $n \in \bbN$, we will generate a grammar $G_n$ with $n$ states
over $\{x, y\}$, two transitions for each state, for which the set of
vertices of the convex hull of
$\prod(G)$ is the set of points
with coordinates $(y(N-y),y)$ for each odd $y \in [0..N]$,
where $N = 2^{n+1}$. It is easy to find $G_1$; we will now show how to
construct $G_{n+1}$ using $G_n$. We perform the following steps.

\begin{itemize}
\item We add $x^{N^2}$ to be always generated right away from the start
symbol $S$ (i.e., to both transitions from $S$).

\item Whenever we generate an $y$ using some transition, we
additionally generate $x^{-N}$, ignoring (for now) the fact that $-N$ is
negative. 
Our vertices are now $(f(y), y)$, where $f(y) = 
y(N-y-N)+N^2$ for
$y$ as before. Note that $f(y) = (N-y)(N+y)$.

\item We replace all occurences of $y$ in our grammar
with a new symbol $Y$, with two rules:
$Y \ra y | y^{-1}$. Our vertices are now still $(f(y), y)$,
except that now $y$ is now in range $[-N..N]$.

\item We add $y^N$ to be always generated right away from $S$.
Now, our vertices are $(g(y), y)$, for each odd $y$ in $[0..2N]$,
and $g(y) = f(y-N) = (2N-y)y$.

\item Thus, we have $G_{n+1}$, except that our grammar is improper due to
negative transitions.
However, it is easy to ``normalize'' our grammar:
if it is possible to generate, say, $x^{-a}$ from a 
non-initial state $A$,
add $x^a$ to the right side of each transition from $A$ (thus eliminating
$x^{-a}$), and replace each occurence of $A$ on the right side of some
transition with $Ax^{-a}$. Since the grammar is acyclic, and the initial
state $S$ never generates a negative number of any terminal,
this algorithm will eventually eliminate all the negative transitions.\qed
\end{itemize}
\end{proof}

\section{Inclusion for fixed alphabets over more than 2 letters}

The proof of the generalization Theorem \ref{maintheorem} to alphabets of
larger (but still fixed) size is very long and technical. We had to 
omit most proofs for space reasons.

In this proof, we will require lots of constants; some of them are
dependant on other. To keep our constants ordered, and make sure that
there is no circular reference between them, we will name them
consistently $C_1, C_2, C_3, \ldots$ through the whole section; each constant
will be defined in such a way that it will depend single exponentially on the
size of the grammar and/or polynomially on the lower numbered constants.
By induction, all numbered constants depend single exponentially on the size of
the grammar. As usual, when we say \emph{$C_i$ is polynomial in $C_j$},
we assume that the size of alphabet $\alph$ is fixed. (If $\alphs$ is not
fixed, then $C_i = O(C_j^{p(\alphs)})$, where $p$ is a polynomial.)

By $\osum A B$ we denote $\{\sum_{b \in B} \alpha_b b : \forall b\in B \ \alpha_b \in A\}$.

Let $J_\alph = \{x \in \bbR^\alph: x \geq 0, |x| = 1\}$.

Let $F_\maxcycle \subseteq [0..\macco]^1_\alph$ (i.e., a set of some 
linear functions over $\bbN^\alph$ with integer coefficients up to $\macco$)
be such that for each set of
$\alphs-1$ vertices $V \subseteq [0..\maxcycle]^\alph$, there exists a 
non-zero $f \in F_\maxcycle$
such that $fV = 0$. This can be done with $\macco$ polynomial in $\maxcycle$.

Let $L_\regval = \{0, \regval\}^\alph$ be the set of vertices of the hypercube of
dimension $\alphs$ and edge length $\regval$.

Let $\regf(\maxcycle,\regval)$ be the set of functions from $F_\maxcycle \times L_\regval$ to $\{-1,0,1\}$.

For a $r \in \regf$, let
\begin{eqnarray*}
\reg(r) &=& \left\{v \in \bbN^\alph: \begin{array}{l}\forall f\in F_\maxcycle \forall l\in L_\regval \\ \sgn (fv-fl) = r_{f,l} \end{array} \right\}, \\
\Reg(r) &=& \left\{v \in \bbR^\alph: \begin{array}{l}\forall f\in F_\maxcycle \forall l\in L_\regval \\ \sgn (fv-fl) \in \{0, r_{f,l}\} \end{array} \right\}, \\
\tau(r) &=& \left\{x \in \bbR^\alph: \begin{array}{l}x \geq 0, x \neq 0, \\ \forall f\in F_\maxcycle \forall l\in L_\regval \\ \sgn (fx) \in \{0,r_{f,l}\} \end{array} \right\}.
\end{eqnarray*}

The following picture (Figure A) shows what $J_\alph$ and $\tau(r) \cap J_\alph$ look
like for $\maxcycle = 2$ and $\alphs = 3$. (If $t \in \tau(r),$ then also
$xt \in \tau(r)$ for $x>0$; thus, a cross of $\tau(r)$ and $J_\alph$ gives
us information about the whole $\tau(r)$.)

\begin{center}
\includegraphics{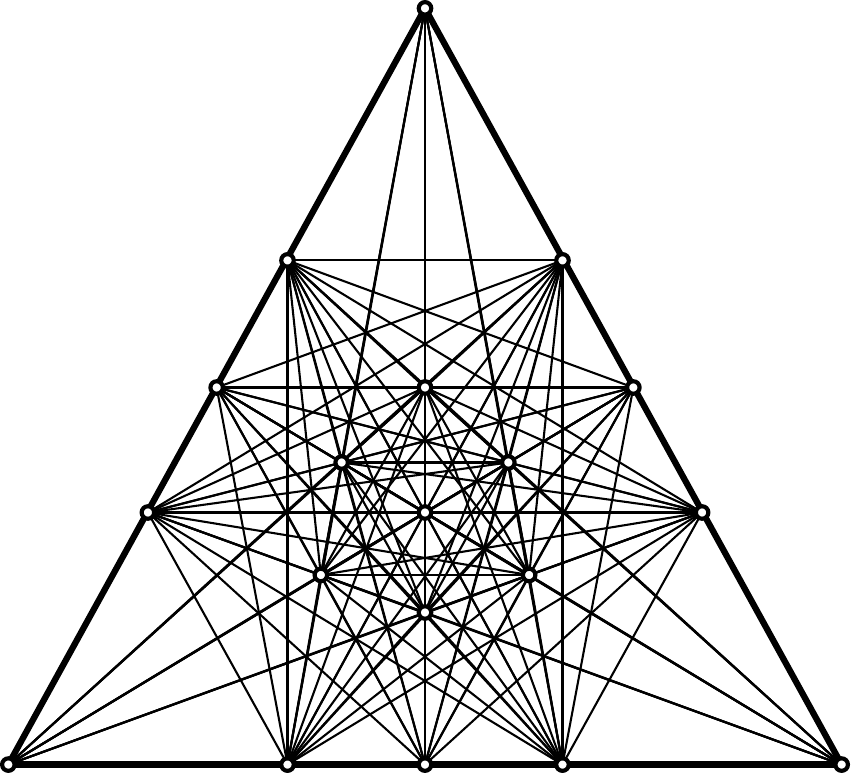}

Figure A
\end{center}

The big equilateral triangle is $J_\alph$. The 19 small white circles are
points $v/|v|$ for $v \in [0..\maxcycle]^\alph$. We connect each pair of
points with a line; these lines correspond to elements of $F_\maxcycle$.

For each $r$, $\tau(r) \cap J_\alph$ is a part of the triangle defined
by their relationship with each line (above, below, or on the line).
Thus, each $\tau(r)$ is either an empty set, or one of the points where
lines cross (including the 19 circles), or a line segment between two
consecutive points where the lines cross, or a polygon bounded by lines.

What does the subdivision of $\bbN^\alph \cap K J_\alph$ into regions 
($\reg$ and $\Reg$) look like for a large $K$? The picture would be similar
to the picture of $\tau(r) \cap J_\alph$, except that instead of each line we
would have a bundle of lines corresponding to picking different elements of
$L_\regval$. In case of $|\alph|=2$ the subdivision of $\bbN^\alph$ into
regions is similar to the picture from Lemma \ref{anglelemma} (the bundles of
lines are no longer infinite).

\begin{lemma}\label{subtraction}
Let $\maxcycle, \regval, \cyclemul \in \bbN$.
There exist constants $\reggen$ polynomial in 
$\maxcycle$, and $\regorder$ polynomial in $\maxcycle$, $\regval$ and
$\cyclemul$, such
that for each $r \in \regf(\maxcycle,\regval)$, for each $v \in \reg(r)$,
if $||v|| \geq \regorder$, then $v = v_0 + \cyclemul t$, where
$t \in \tau(r) \cap [0..\reggen]^\alph$ and $v_0 \in \reg(r)$.
\end{lemma}

Note that a $\alphs$-dimensional version of Lemma \ref{anglelemma} follows
easily from Lemma \ref{subtraction}. Lemma \ref{subtraction} also plays a
similar role in our proof as Lemma \ref{anglelemma} did for $d=2$.

\begin{lemma}\label{separation}
Let $\reggen \in \bbN$. Then there exists $\csep$ polynomial in $\reggen$ such that:

Let $P, Q \subseteq [0..\reggen]^\alph$ such that
$\osum\bbP P \cap \osum\bbP Q = \{0\}$, and $0 \notin P,Q$.
Then there is a $\Phi \in \bbN^1_\alph$ 
(i.e.,~a linear function over $\bbN^\Sigma$ with integer coefficients)
such that $\Phi P > 0$, $\Phi Q < 0$,
and $|\Phi| \leq \csep$.
\end{lemma}

Intuitively, this lemma states that, given two disjoint closed convex polygons
in some space (in our case, the space is $J_\alph$, and the polygons are
intersections with 
$\osum \bbP P$ and $\osum \bbP Q$), we can separate them strictly with a
hyperplane. Such separation
is a well known property of convex sets; Lemma \ref{separation} gives a
polynomial bound on the coefficients of such a separating hyperplane.

\begin{lemma}\label{matrixform}
Let $\maxcycle, \maxrun \in \bbN$. Then there exists a constant
$\bigconst$ such that:

Let $S = W + \osum\bbN \OCyc$, where $W \subseteq [0..\maxrun]^\alph$ and
$\OCyc \subseteq [0..\maxcycle]^\alph$. 
Let $r \in \regf(\maxcycle,\regval)$.
Then there exists a matrix $M \in [0..\maxcycle]^\alph_\alph$
such that $S \cap \reg(r) = (W_1 + M \bbN^\alph) \cap \reg(r)$, where $W_1 \subseteq [0..\bigconst]^\alph$.
\end{lemma}

\begin{theorem}\label{frameregion}
Let $G$ be a commutative grammar, and let $r$ be a region.

The intersection of $\prod(G) \cap \reg(r)$ is an intersection of $\reg(r)$
and a polynomial union of $\bigconst,\maxcycle$-frames, where
$\bigconst$ and $\maxcycle$ are single polynomial in $|G|$.
\end{theorem}

\begin{theorem}\label{equalthree}
Inclusion is \PiP-complete for fixed $\alph$.
\end{theorem}

\begin{proof}[ of Theorem \ref{equalthree}]
We apply the methods of Theorem \ref{maintheorem} separately for each region from
Theorem \ref{frameregion}.
\qed\end{proof}

\section{Lower bounds}

For completeness, we provide proofs of lower bounds for the complexities of
considered problems. These results have been previously known (e.g.,
\cite{hyunh1}).

\begin{theorem}\label{membunfix}
For a commutative regular grammar 
$G$ over $\alph$ (whose size is not fixed), and $K \in \mult \alph$,
the problem of deciding whether $K \in \prod(G)$ is NP-complete.
\end{theorem}

\begin{proof}
The problem is obviously in NP (the run is the witness --- the only
problem is that it could be larger than polynomial by including
a large number of transitions which produce nothing, but such transitions
must form cycles which can be easily eliminated). We show a reduction from
the Hamiltonian
circuit problem. Let $(V,E)$ be a graph. We take $\alph = S = V$, and 
for each edge $(v_1, v_2)$ we add a transition $\mv{v_1}{v_2}{v_2}$.
We pick an initial state $s_0$ and add a final transition $\mv{s_0}{0}{0}$.
The graph $(V,E)$ has a Hamiltonian circuit iff $(1,1,\ldots) \in \prod(G)$.
\qed\end{proof}

The same example shows that disjointness is co-NP-hard for grammars over
alphabets of unfixed size. It is also co-NP-complete, since our polynomial
algorithm for fixed size alphabets can be easily modified to work in co-NP 
for unfixed size ones.

\begin{theorem}\label{membnonreg}
For a single letter alphabet $\alph = \{a\}$ and a commutative grammar 
$G$ (not necessarily regular), and $K \in \mult \alph$,
the problem of deciding whether $K \in \prod(G)$ is NP-complete.
\end{theorem}

\begin{proof}
The problem is in NP for a similar reason. We can reduce the knapsack
problem: given a sequence of positive integers $k_1, \ldots, k_m$ and
$K$, is there a subset $L \subseteq \{1, \ldots, m\}$ such that
$\sum_{l \in L} k_l = K$? Indeed, it is easy to produce a grammar
of size $O(\sum_i \log(k_i))$ which generates $Ka$ iff $K = \sum_{l \in L} k_l$
for some $L$.
\qed\end{proof}

\begin{theorem}\label{uninonreg}
Let $G$ be a commutative regular grammar over $\alph$ of fixed size.
Then the problem of deciding universality ($\prod(G) = \bbN^\alph$)
is coNP-hard even for $\alph = \{a\}$. 
\end{theorem}

\begin{proof}
The problem is in coNP because the witness for non-universality is
of polynomial length (by the same argument as in Theorem \ref{maintheorem}).

We reduce the 3CNF-SAT problem. Let $\phi = \bigwedge_{1\leq i\leq k} C_i$ be a
3CNF-formula with $n$ variables $x_1\ldots x_n$ (which can be 0 or 1)
and $k$ clauses. Let $p_1, p_2, \ldots, p_n$
be $n$ distinct prime numbers. Let $i \in [1..k]$. Suppose that clause $C_i$
is of form $\bigvee_{k\in[1..3]} x_{a_k}=v_{a_k}$. Our grammar will have states
$S^i_j$, where $0 \leq j < M_i = p_{a_1} p_{a_2} p_{a_3}$; we have cyclic
transitions $\mv{S^i_j}{a}{S^i_{(j+1) mod M_i}}$, and $\mv{S^i_j}{}{0}$ for each 
$j$ satisfying $\bigvee_{k\in[1..3]} j \bmod p_{a_k}=v_{a_k}$. We also have 
transitions $\mv{s_0}{}{S^i_0}$ for each $i$.

From simple number theoretic arguments we get
that $x \notin \prod(G)$ iff the formula $\phi$ is satisfied for $x_i = x
\bmod p_i$.\qed
\qed\end{proof}

\begin{corol}
Disjointness is coNP-complete for commutative context-free grammars
over $\alph$ of fixed size, and universality, equivalence, and inclusion
are coNP-complete for commutative regular grammars over $\alph$ of fixed
size.
\end{corol}

\begin{proof}
We get that disjointness and universality for commutative grammars are
coNP-hard from Theorems \ref{membnonreg} and \ref{uninonreg}, respectively.
We get the upper bounds by applying the same methods as in Theorem
\ref{equalthree} (or the easier Theorem
\ref{maintheorem} for alphabets of size 2).
In the case of equality and inclusion for regular grammars, we get rid of one
level of the polynomial hierarchy by using Theorem \ref{algoregmem} to
decide membership.
\qed\end{proof}

\section{Conclusion}

We have shown tight complexity bounds for the problems of membership,
inclusion (equality), and disjointness of Parikh images of
regular and context-free languages over alphabets of fixed size.

What about alphabets of unbounded size? Some of the problems here remain
open; we do not know whether our results and methods shed much light on
these problems. For example, as far as we know, equality of Parikh images
of both regular and context-free commutative languages (over alphabets of
unfixed size) is only known to be \PiP-hard and in coNEXPTIME \cite{semipi2}.
Also, for the universality problem for commutative context-free grammars
over alphabets of fixed size, our bounds are not tight: we know
that this problem is in \PiP (as a special case of inclusion), but the only
lower bound known to us is coNP (from the regular version).

In some places in our paper, it was convenient to use grammars which could
produce negative quantities of letters (Theorem \ref{algoregdis}), or even
negative quantities of states (Theorem \ref{counterthree}). It is interesting
whether there exists some more general theory for such techniques.

Many thanks to Sławek Lasota for introducing me to these problems, and to
everyone on our Automata Scientific Excursion for the great atmosphere of
research.

\twocolumn

\appendix

\section{Equality of context-free grammars for a fixed $d>2$ --- proof details}

\begin{proof}[ of Lemma \ref{subtraction}]
Let $r \in \regf$.

For all regions $\reg(r)$ bounded by some $M$ (which must be polynomial),
we can take $\regorder = M, \reggen = 0$.

Let $v_i$ be a sequence of elements of $\reg(r)$ such that $\lim_{i\ra \infty} |v_i| = \infty$.
Let $w_i = v_i / |v_i|$. Let $w$ be a cluster point of $w_i$. We have
$w \in \tau(r)$. Thus, $\tau(r) \cap J_\alph$ is non-empty.

Since $\tau(r)$ is given by linear equations, we get that $\tau(r) \cap J_\alph$
is a polytope whose vertices (black points where edge cross in Figure A)
are $t_1, \ldots, t_D$, where
$t_i \in [0..\rgden]^\alph / q_i$, where $q_i \in [1..\rgden]$. 
Both $\rgden$ and $D$ are bounded polynomially.

Let 
$R^*$ be the set of points $v \in \Reg(r)$ which cannot be written
as $v_0+t$, where $t \in \tau(r)$ and $v_0$ is also in $\Reg(r)$. 
It can be easily seen that $\Reg(r) = R^* + \tau(r)$, 
and also that $R^*$ is bounded polynomially by $M$.
                                                                    
We will show that our claim is satisfied for $\reggen = \rgden^2$ and
$\regorder = M + \cyclemul \rgden D$.
Let $v \in \reg(r), ||v|| \geq \regorder$.
Since $v \in \reg(r) \subseteq \Reg(r)$, 
we have $v = v_{00} + \sum_i \alpha_i t_i$, where $\alpha_i \geq 0$ and
$v_{00} \in R^*$.

Since $||v|| > \regorder$, there must be $i$ such that
$\alpha_i > \cyclemul \rgden$. Thus, $\alpha_i > \cyclemul q_i$.
We get our form: $v = v_0 + \cyclemul q_i t_i$, where $v_0$ is also in
$\reg(r)$.
\qed\end{proof}

\begin{proof}[ of Lemma \ref{separation}]
Let $||x||_2$ denote the Euclidean norm of $x \in \bbR^\Sigma$.
Let $X = (J_\Sigma \cap \oplus\bbP P) - (J_\Sigma \cap \oplus\bbP Q)$;
since $\oplus\bbP P$ and $\oplus\bbP Q$ are disjoint, $0 \neq X$.
Let $x \in X$ be the point of $x$ such that
$||x||_2 = \min \{||x||_2: x \in X\}$.

Now, let $x_p, x_q \in J_\Sigma$ be the points such that $x_p - x_q = x$.
Let $\Phi \in \bbP_\Sigma^1$ be such that $\Phi(x_p) = 1$, $\Phi(x_q) = -1$,
and $\Phi(z) = 0$ for all $z \in J_\Sigma$ such that $||z-x_p||_2 = ||z-x_q||_2$.

This $\Phi$ satisfies our conditions. We omit the proof that this construction
indeed works, and that $||\Phi||$ is bounded polynomially.

\end{proof}

\begin{proof}[ of Lemma \ref{matrixform}]

Let $\cyclemul$ be the bound on $|\det M|$ for $M \in [0..\maxcycle]^\alph_\alph$.

Let $\reggen$ and $\regorder$ be from Lemma \ref{subtraction} (for our
$\maxcycle$, $\regval$, and $\cyclemul$).

Let $\csep$ be from Lemma \ref{separation} (for our $\reggen$).

Let $\mulgen$ be such that for each $v \in [0..\csep]^\alph$, and each
$P \subseteq [0..\maxcycle]^\alph$, if $v \in \osum\bbP P$, then $Xv \in \osum{[0..\mulgen]} P$
for some $X \in [0..\cyclemul]$.

Let $\bigconst$ be big enough.

Let $r \in \regf(\maxcycle,\regval)$. For $r$ such that $\reg(r)$ is bounded
(by $M_0$),
$M = 0$ and $\bigconst = M_0$. Thus, assume then $\reg(r)$ is unbounded.

Let $H = \{w/|w|: w \in [0..\maxcycle]^\alph\}$.
$\OCyc \bbP^\alph \cap J_\alph$ is a convex polytope with vertices from $H$.
On the other hand, $\tau(r) \cap J_\alph$ is a convex polytope bounded by hyperplanes
going through sets of $\alphs-1$ vertices from $H$; moreover, it is a
minimal such polytope, i.e.,~it cannot be subdivided into two such polytopes
of the same dimension by such a hyperplane. Thus, either
$\int(\tau(r)) \cap J_\alph$
is disjoint with $\OCyc \bbP^\alph \cap J_\alph$ (case 1),
or $\tau(r) \cap J_\alph$
is a subset of $\OCyc \bbP^\alph \cap J_\alph$ (case 2).

In the case (1), there must be a hyperplane separating $\tau(r) \cap J_\alph$
and $\OCyc \bbP^\alph \cap J_\alph$. Let $f \in F_\maxcycle$ be such that
$f(\tau(r)) \geq 0$, $f(t) > 0$ for some $t \in \tau(r)$, and
$f(\OCyc) \leq 0$. Since $f(t) > 0$, we have $r(f,l) = 1$ for all
$l \in L_\regval$. Thus, for $v \in S \cap \reg(r)$, we have 
$f(v) > f(w_0)$ for all $w_0 \in [0..\regval]^\alph$. On the other hand,
for some $w_0$ we have $v = w_0 + \sum_{Y \in \OCyc} n_Y Y,$ thus
$f(v) \leq f(w_0)$. A contradiction. Thus, $S \cap \reg(r) = \emptyset$.

In the case (2), there must be a matrix $M$, whose columns are
$\alphs$ elements of $\OCyc$, such that $\tau(r) \subseteq M \bbP^\alph$.

Let $v \in S \cap \reg(r)$. We prove inductively by $v$.

If $||v|| \leq \bigconst$, we are ready.

Otherwise, using Lemma \ref{subtraction} iteratively, we write
$v$ as $v_0 + \cyclemul (t_1 + \ldots + t_K)$, where $t_i \in [0..\reggen]^\alph$, and
$||v_0|| \leq \regorder$. We have $K > (\bigconst-\regorder) / \cyclemul\reggen$.

On the other hand, we can write $v$ as $w_0 + \sum_{Y \in \OCyc} n_Y Y
+ \sum_{Y \in P} m_Y Y$, where $w_0 \in [0..\maxrun]^\alph$, $n_Y \leq \mulgen$,
$m_Y \geq \mulgen$, $|P| \leq \alphs$. (We get this form just like in the proof
of Theorem \ref{theoreg}.)

If for some $i$ we have $t_i \in \osum\bbP P$, then we are done. Indeed,
from
definition of $\mulgen$ we have that $\cyclemul t_i = \sum_{Y \in P} \alpha_Y Y$,
where $\alpha_Y < \mulgen$. On the other hand, $\cyclemul t_i = \sum M^i \beta_i$,
$\beta_i \in \bbN$.
From the induction hypothesis we can present $v-\cyclemul t_i$ in our form $F$.
Thus we can also present $v$ as $F + \sum_i M^i \beta_i$.

Now, what if $t_i \notin \osum\bbP P$? From Lemma \ref{separation}, let
$\Phi$ be such that $\Phi(t_i) < 0$, $\Phi(P) > 0$, $||\Phi|| < \csep$.
We have:

\begin{eqnarray*}
\Phi(v) &=& \Phi(v_0) + \cyclemul  \sum_i \Phi(t_i) < \regorder \csep - \cyclemul K \\
\Phi(v) &=& \Phi(w_0) + \sum_{Y\in \OCyc} n_Y Y + \sum_{Y\in P} m_Y Y > \\
&&
- \maxrun \csep d - |\OCyc| \mulgen \maxcycle \csep d + \\ &&
+ ((\bigconst - \maxrun + |\OCyc| \mulgen \maxcycle) / \maxcycle)
\end{eqnarray*}

This is a contradiction for $\bigconst$ big enough.

\qed\end{proof}

\begin{proof}[ of Theorem \ref{frameregion}]

Let $\maxcycle$ be the bound on the size of a simple cycle,
i.e.,~$\OCyc_S \subseteq [0..\maxcycle]^\alph$ (Lemma \ref{exponlem}).

Let $\maxrunp$ be such that for each run $D$ we have 
$\prod(D) = \prod(D_0) + \sum_{Y \in P} n_Y Y$, where $||D_0|| \leq \maxrunp$,
and $P$ is a subset of $\OCyc_{\supp D}$ of size $\alphs$. (We get this
form and a polynomial bound for $\maxrunp$ just like in the proof of Theorem
\ref{theoreg}.)

Let $I$ be the set of
all subsets of $S$ containing at most $\alphs$ elements. 
For $i \in I$, we can create
$W_i \subseteq [0..\maxrun]^\alph$ so that
$\prod(G) = \bigcup_{i\in I} W_i + \osum \bbN \OCyc_i$.
The method is similar to the one used in the proof of Theorem
\ref{grammarnormalform}.

Use constants just like in Lemma \ref{matrixform}.

Applying Lemma \ref{matrixform} to each component of the union, we get
that for each $r$,
$\reg(r) \cap \prod(G) = \reg(r) \cap \bigcup_{i \in I} W_i' + M_i \bbN^\alph$.
\qed\end{proof}

\end{document}